\documentclass{amsart}
\usepackage{amsmath,amssymb,amscd,amsfonts}
\usepackage{graphicx}

\theoremstyle{plain}
\newtheorem{thm}{Theorem}[section]

\newtheorem{prop}[thm]{Proposition}

\theoremstyle{definition}
\newtheorem{dfn}[thm]{Definition}

\newtheorem{problem}[thm]{Problem}

\theoremstyle{remark}
\newtheorem{rem}[thm]{Remark}

\numberwithin{equation}{section}

\begin{document}
\title[On the existence of infinitely many universal tree-based networks] 
{On the existence of infinitely many universal tree-based networks}
\author{Momoko Hayamizu$^{\ast\dagger}$}
\thanks{*~Department of Statistical Science, The Graduate University of Advanced Studies}
\thanks{\dag~The Institute of Statistical Mathematics}

\address[Momoko Hayamizu]{The Institute of Statistical Mathematics\\
 10-3 Midori-cho\\ Tachikawa, Tokyo 190-8562\\ Japan}
\email[]{hayamizu@ism.ac.jp}

\maketitle

\begin{abstract}
A tree-based network  on a set $X$ of $n$ leaves is said to be  \emph{universal} if any rooted binary phylogenetic tree on $X$ can be its base tree. 
Francis and Steel showed that there is a universal tree-based network on $X$ in the case of $n=3$, and asked whether such a network exists in general. We settle this problem by proving that there are infinitely many universal tree-based networks for any $n>1$.
\end{abstract}

\section{Introduction}\label{Introduction}
Throughout this paper,   $n$ denotes  a natural number that is greater than 1 and $X$ represents the set $\{1, 2, \cdots,n\}$.  All graphs considered here are  directed acyclic graphs. A graph $G^\prime$ is said to be a  \emph{subdivision} of a graph $G$ if $G^\prime$ can be obtained from $G$ by inserting vertices into arcs of $G$ zero or more times. Given a vertex  $v$ of a graph with $indeg(v)=outdeg(v)=1$,  \emph{smoothing} (or \emph{supressing}) $v$ refers to removing $v$ and then adding an arc from the parent to the child of $v$. Two graphs are said to be \emph{homeomorphic} if they become isomorphic after smoothing all vertices of in-degree one and out-degree one. 

For the reader's convenience, we briefly recall the relevant background from~\cite{FS} (see~\cite{SS} for the terminology in phylogenetics).

\begin{dfn}\label{BPN}
A \emph{rooted binary phylogenetic network on $X$} is defined to be a directed acyclic graph $(V,A)$ with the following properties: 

\begin{itemize}
\item $X=\{v\in V\,|\,indeg(v)=1,\, outdeg(v)=0\}$; 
\item there is a unique vertex $\rho\in V$ with $indeg(\rho)=0$ and $outdeg(\rho)\in\{1,2\}$; 
\item for all $v\in V\setminus\{X\cup \{\rho\}\}$, $\{indeg(v), outdeg(v)\}=\{1,2\}$. 
\end{itemize}
The vertices in $X$ are called \emph{leaves}, and the vertex $\rho$ is called the \emph{root}.
\end{dfn}

\begin{dfn}
Suppose $\mathcal{T}=(V,A)$ is a rooted binary phylogenetic tree on $X$. 
A rooted binary phylogenetic network $\mathcal{N}$ on $X$ is said to be a \emph{tree-based network on $X$ with base tree $\mathcal{T}$} if there are a subdivision $\mathcal{T^\prime}=(V^\prime,A^\prime)$ of  $\mathcal{T}$  and a set $I$ of  mutually vertex-disjoint arcs between  vertices in $V^\prime\setminus V$ such that $(V^\prime, A^\prime\cup I)$ is acyclic and  is homeomorphic to $\mathcal{N}$. 
The vertices in $V^\prime\setminus V$ are called  \emph{attachment points}, and the arcs in $I$ are called \emph{linking arcs}.
\end{dfn}

Tree-based networks can have an important role to play in modern phylogenetic inference as they can represent more intricate or realistic relationships among taxa than phylogenetic trees without compromising the concept of `underlying trees' (\textit{cf.},~\cite{FS, Huson}).

In order to state the problem formally, we now introduce the notion of universal tree-based networks. A tree-based network on $X$ is said to be \emph{universal} if any binary phylogenetic tree on $X$ can be a base tree. We can define universal tree-based networks in a more concrete manner with the number $(2n-3)!!$ of binary phylogenetic trees on $X$ as follows.

\begin{dfn}\label{universalTBN}
A tree-based network $\mathcal{N}=(V,A)$ on $X$ is said to be \emph{universal} if for any binary phylogenetic tree $\mathcal{T}^{(i)}$ on $X$ ($i\in\{1, 2, \cdots, (2n-3)!!\}$), there is a set  $I^{(i)}\subset A$ of linking  arcs such that $(V,A\setminus I^{(i)})$ is homeomorphic to $\mathcal{T}^{(i)}$.
\end{dfn}

\begin{problem}[\cite{FS}]\label{problem}
Does a universal tree-based network on a set $X$ of $n$ leaves exist for all $n$?
\end{problem}
 
This problem is fundamental because it explores whether a phylogenetic tree on $X$ is always reconstructable from a tree-based network on $X$. In~\cite{FS}, Francis and Steel pointed out that the answer is `yes' for $n=3$.  
In this paper, we will completely settle their question in the affirmative and provide further insights into universal tree-based networks (Theorem~\ref{thm}).

\section{Preliminaries}\label{Preliminaries}
Here, we slightly generalise the concept of tree shapes.  
Given a tree-based network $\mathcal{N}$ on $X$, ignoring the labels on the leaves of $\mathcal{N}$ results in an unlabelled tree-based network $N$ with $n$ leaves. We use the two different types of symbols, such as $N$ and $\mathcal{N}$, to mean unlabelled and labelled tree-based networks, respectively.  Two tree-based networks $\mathcal{N}$ and $\mathcal{N}^\prime$ on $X$ are said to be \emph{shape equivalent}  if $N$ and $N^\prime$ are isomorphic. This equivalence relation partitions a set of the tree-based networks on $X$ into equivalence classes called \emph{tree-based network shapes with $n$ leaves}.

\begin{dfn}[]\label{universalTBNshape}
A tree-based network shape $N$ with $n$ leaves is said to be \emph{universal}  if for any rooted binary phylogenetic tree shape $T^{(i)}$ with $n$ leaves ($i\in\{1,2,\cdots,r_n\}$), there is a set $I^{(i)}$ of linking arcs such that $(V,A\setminus I^{(i)})$ is homeomorphic to $T^{(i)}$. Here, $r_n$ denotes the  number of rooted binary phylogenetic tree shapes with $n$ leaves. 
\end{dfn}

The following proposition is not directly relevant to this paper, but  ideas behind it, which are summarised in Remark~\ref{rem}, will be useful in the proof of Theorem~\ref{thm}. 

\begin{prop}[\cite{Harding}]\label{Harding}
Let  $r_1:=1$ and $k\in\mathbb{N}$ with $k>1$. Then, we have the following recurrence equation:  
\begin{eqnarray*}
r_n=\left\{ \begin{array}{ll}
1 & \text{if $n=2$;} \\
\sum_{i=1}^{k-1}r_i r_{n-i} & \text{if $n=2k-1$;} \\
\frac{r_k (r_k+1)}{2}+\sum_{i=1}^{k-1}r_i r_{n-i} & \text{if $n=2k$.} \\
\end{array} \right.
\end{eqnarray*}
\end{prop}

\begin{rem}\label{rem}
We assume that $T_1$ represents a rooted chain shape. Any  rooted binary phylogenetic tree shape $T_n$ with $n$ leaves can be decomposed into two first-order subshapes $T_m$ and $T_{n-m}$   with $m\in\mathbb{N}$. In other words, using Harding's notation~\cite{Harding}, we can write $T_n = T_m + T_{n-m}$. 
\end{rem}

\section{Results}\label{Results}
\begin{thm}[]\label{thm}
For any natural number $n>1$, there are infinitely many universal tree-based networks on a set $X$ of $n$ leaves.
\end{thm}

\begin{proof}
First, we will show that there is a universal tree-based network shape with $n$ leaves. Let $U_n$ be a rooted binary phylogenetic network shape with $n$ leaves as illustrated in the left panel of Figure 1, which can be obtained  by adding $(n-1)(n-2)/2$ linking arcs and $(n-1)(n-2)$ attachment points to a rooted caterpillar tree shape with $n$ leaves. By definition, $U_n$ is a tree-based network shape with $n$ leaves. We will prove that $U_n$ is universal by induction. (i) It is easy to see that $U_2$ and $U_3$ are universal.  (ii) Assuming $U_k$ is universal for any $k\in\mathbb{N}$ ($2\leq k\leq n$), we will show that $U_{n+1}$ is universal. We claim that any binary phylogenetic tree shape $T_{n+1}$ with $T_{n+1}=T_{n}+T_{1}$ can be a base tree shape of $U_{n+1}$. Indeed,  $U_{n+1}$ contains  mutually vertex-disjoint arcs  whose removal turns $U_{n+1}$ into the union of two subgraphs that are homeomorphic to $U_n$ and $T_1$, respectively (see the middle panel of Figure~\ref{Ushape}). Because $U_n$ is universal, our claim holds true. 
We next claim that any binary phylogenetic tree shape $T_{n+1}$ with $T_{n+1}=T_{k}+T_{n-k+1}$ can be a base tree of $U_{n+1}$. The right panel of Figure~\ref{Ushape} indicates that $U_{n+1}$ contains two distinct subsets of mutually vertex-disjoint arcs, one of which delineates $U_{n-k+1}$ (shown in thick gray line) and the other distinguishes $U_k$ from the remainder. Because both $U_{n-k+1}$ and $U_k$ are universal, our claim holds true.  Therefore, $U_{n+1}$ is universal. Hence, $U_n$ is universal for all $n$.

Next, we will provide a method to create infinitely many universal tree-based networks on $X$ from $U_n$. 
Let $\mathcal{E}_n$ be a tree-based network on $X$ obtained from $U_n$ by specifying a permutation $\pi_0$ of $X$.  In what follows, we use the same notation $i$ both for a leaf labelled  $i$ and for the terminal arc incident with $i$. A \emph{crossover} $\sigma_{ij}$  refers to a pair of crossed additional arcs between two distinct terminal arcs $i$ and $j$ as described in Figure~\ref{crossover}. Note that $\sigma_{ij}$ can be viewed as representing the transposition $(i,\,j)$ of the labels. For any permutation $\pi_1$ ($\neq\pi_0$) of $X$, there is a series  of adjacent crossovers that converts $\pi_0$ into $\pi_1$ and then vice versa (note that any permutation can be expressed as a product of transpositions and that the symmetric group $S_n$ is generated by the adjacent transpositions). Then, by sequentially adding  $n!-1$ series of crossovers, we can construct a universal tree-based network $\mathcal{U}_n$ on $X$ from $\mathcal{E}_n$. Moreover, it is possible to create infinitely many universal tree-based networks on $X$ because we may add an arbitrary number of redundant crossovers  among the terminal arcs of $\mathcal{U}_n$.
This completes the proof. 
\end{proof}

\begin{figure}[htbp]
\centering
\includegraphics[width=1\textwidth]{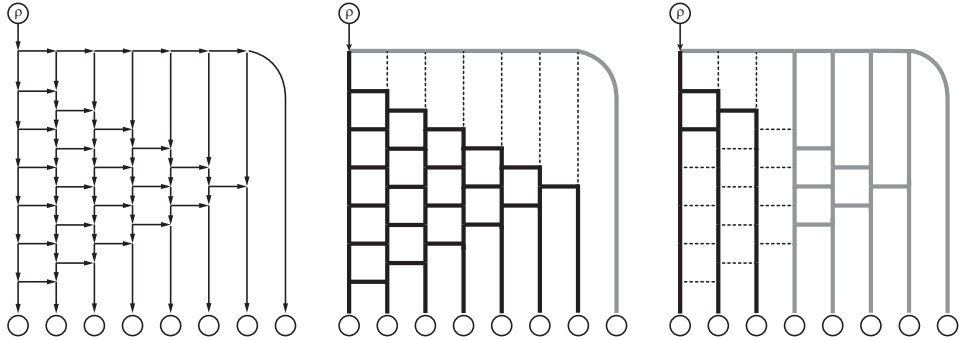}
\caption{The first part of the proof of Theorem~\ref{thm}. 
The left panel is an illustration of $U_n$ for $n=8$. The other panels show examples of $T_{n+1}$ in $U_{n+1}$ for $n+1=8$, and the right panel describes the case of $T_{n+1}=T_k + T_{n-k+1}$ with $k=3$. \label{Ushape}}
\end{figure}

\begin{figure}[htbp]
\centering
\includegraphics[width=0.4\textwidth]{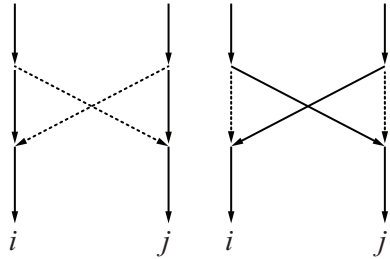}
\caption{The second part of the proof of Theorem~\ref{thm}.  
 Left: A \emph{crossover} $\sigma_{ij}$ is defined to be a pair of crossed additional arcs placed between arcs $i$ and $j$ ($i\neq j$) after subdividing both arcs twice. Right: When the two arcs in $\sigma_{ij}$ are selected as tree arcs, $\sigma_{ij}$ represents the transposition $(i\;j)$. 
  \label{crossover}}
\end{figure}

We note that the construction described in the proof of Theorem~\ref{thm} adds more arcs than necessary (\textit{cf.} Figure 1 in~\cite{FS}). It would be interesting to consider how to construct universal tree-based networks on $X$ with the smallest number of arcs. 

\section*{Comments}
We studied Problem~\ref{problem} independently from Louxin Zhang~\cite{LX}.

\section*{Acknowledgement}
The author thanks Kenji Fukumizu and  the anonymous reviewers for their careful reading and helpful comments. 

\bibliographystyle{amsplain}

\bibliography{universalTBN}

\end{document}